\documentclass[aps,prx,twocolumn]{revtex4-2}
\usepackage[utf8]{inputenc}
\usepackage[T1]{fontenc}
\usepackage{bm}
\usepackage{amsmath, amssymb, amsfonts}  
\usepackage{amsthm}
\usepackage{graphicx}
\usepackage{xcolor}
\usepackage{physics}
\usepackage{caption}
\usepackage{subcaption}
\usepackage{hyperref}
\usepackage{braket}
\usepackage{tikz}

\newtheorem{theorem}{Theorem}[section]

\begin{document}
\title{Geodesic Quantum Walks}

\author{Giuseppe Di Molfetta}
\affiliation{CNRS, LIS, Aix-Marseille Université, Université de Toulon, Marseille, France}
\author{Victor Deng}
\affiliation{Ecole Normale Superieure, Departement d'Informatique, Paris, France}

\begin{abstract}
    We propose a new family of discrete-spacetime quantum walks capable to propagate on any arbitrary triangulations. Moreover we also extend and generalise the duality principle introduced in~\cite{arrighi2019curved}, linking continuous local deformations of a given triangulation and the inhomogeneity of the local unitaries that guide the quantum walker. We proved that in the formal continuous limit, in both space and time, this new family of quantum walks converges to the (1+2)D massless Dirac equation on curved manifolds. We believe that this result has relevance in both modelling/simulating quantum transport on discrete curved structures, such as fullerene molecules or dynamical causal triangulation, and in addressing fast and efficient optimization problems in the context of the curved space optimization methods. 
    \end{abstract}

\maketitle






\textit{Introduction} Classical random walks are well known to approximate Brownian motion and the continuous diffusion equation~\cite{knight1962random}. In 1960 Roberts and Ursell~\cite{roberts1960random} argued that this is still true if the jiggling motion of small dust particles occurs on a curved spacetime instead (such as in the proximities of a strong gravitational field). Later, in 1984, Varopoulos~\cite{varopoulos1984brownian} proved rigorously that, for a compact Riemannian manifold, if a triangulation is given, there always exists a canonical way to define a random walk on the vertices of the triangulation which converges to the Brownian motion. This suggested, more than thirty years ago, that classical random walks on triangulations provide simple procedures to ``discretize” diffusion processes on curved space. The quantum analogue of random walks has also been proven capable of simulating by means of local unitary operations another phenomenon ubiquitous in nature, transport~\cite{arrighi2018dirac, marevs2020quantum}. Quantum Walks describe situations where a quantum particle is taking steps on a lattice conditioned on its internal state, typically a (pseudo) spin one half system. The particle dynamically explores a large Hilbert space associated with the position of the lattice and allows thus to simulate a wide range of physical phenomena. With quantum walks, the transport is driven by an external discrete operation (coin and shift), which sets it apart from other lattice quantum simulation concepts where transport typically rests on tunnelling between adjacent sites: all dynamics processes are discrete in space and time. More in general, they have been extended to graphs~\cite{verga2019interacting} and simplicial complexes~\cite{matsue2016quantum, aristote2020dynamical}. In the same way that classical random walks provide a discretization procedure for the diffusion equation, quantum walks discretize the Dirac equation~\cite{arrighi2014dirac}, i.e. the wave equation that describes the motion of spin one half systems (\textit{aka} fermions) in a relativistically invariant way. There are essentially two ways to discretize transport on curved surfaces through the implementation of quantum walks. The first, more oriented towards quantum simulation on specific physical platforms, is to encode the space-time metrics in local unitaries, thus making them non-homogeneous. This procedure would correspond to sampling from the continuous manifold a finite data set of size corresponding to the number of grid points. Indeed, one of the authors has already proved that quantum walks can efficiently integrate the curved dynamics of a single particle in continuous space-time and converge to the general covariant Dirac equation~\cite{arrighi2018quantum, di2013quantum}. These results were obtained on one- and two-dimensional rectangular grids~\cite{arnault2017quantum} and later extended to any spatial dimensions~\cite{arrighi2017quantum}. The other way is to consider a discrete curved surface, e.g. a non-homogeneous triangulation and define a homogeneous quantum walk on it. This procedure is more oriented towards modelling a transport phenomenon on discrete curved structures such as fullerene molecules, which consist entirely of carbon and take the form of a hollow sphere, ellipsoid or tubular. 


\textit{Related work} The aforementioned two approaches, long separated, have recently been reconciled by one of the authors, proving that there is, in some specific cases, a duality between the continuous local deformation of a triangulation and the inhomogeneity of the local unitaries that guide the quantum walker~\cite{arrighi2019curved}. However, it was possible to prove this duality principle only in the case of diagonal deformation matrices, which in the case of the space-time tensor, would correspond to synchronous metrics. In this manuscript we propose to extend this result to any non-diagonal metric and to valorise the previously introduced duality principle in its maximum generality. Moreover, our work recalls the so-called Geodesic Random Walk, first introduced by Jørgensen~\cite{jorgensen1975central} and recently extended to Finsler manifolds~\cite{ma2021geodesic}. In such classical walks, in the absence of coin-degree of freedom, the spacetime metrics coefficients are embedded in the spatial increments. 

\textit{Motivations} There are numerous motivations to introduce the Geodesic Quantum Walk~(GQW). First is the emergence of massless Dirac fermions on graphene-like materials~\cite{novoselov2005two}, and within crystals in general. Quantum transport within such materials may be the physical phenomena that we wish to model by GQWs. Another topic is related to topological states, which are well known to be non-trivial on triangular grids~\cite{kitagawa2010exploring}. Moreover GQWs should allow us to model all sorts of topologies as simplicial complexes, and the duality principle would provide a simple procedure for their quantum simulation. Finally yet another motivation for exploring non-flat geometries is general relativity. Simulating curved transport on spacetime triangulation is reminiscent of the question of matter propagation in triangulated spacetime, as arising, e.g., in Causal Dynamical Triangulation~\cite{ambjorn2014quantum} or Loop Quantum Gravity~\cite{rovelli2015covariant}. Finally, irregular lattices and/or random graphs may also be of interest: can they be re-interpreted in geometrical terms, i.e. in terms of an effective metric? Let us precise that a rigorous classical limit of these quantum schemes, including the quantum gravity theories, is still lacking and that techniques such as coarse graining quantum maps~\cite{duranthon2021coarse}, in the context of quantum automata, would play a central role in order to derive the macroscopic classical dynamical equations in the thermodynamic limit. 

The manuscript is organised as follows: In Sec.~\ref{subsect:triangular_qw} we introduce the basic Quantum Walk over equilateral triangles and we recall how to derive the formal continuous limit in space and time. Then Sec.~\ref{sec:geoQW} is dedicated to extending the definition of the QW in a way that the local operators could take into account any arbitrary triangulations; this results in covering transport over curved null-geodesics and we thus prove that, taking the formal continuous limit, we recover a transport equation in curved spacetime, namely the general covariant Dirac equation. In Sec.~\ref{sec:conc} we provide a summary and some perspectives.


\section{Quantum Walking over triangles}
\label{subsect:triangular_qw}

Here we first reviewed the uniform QW as introduced by one of the authors in \cite{arrighi2018dirac}. The walker is defined on the edges of a regular triangular grid, which sets it apart from other QWs defined on triangles and honeycomb structures, where the QW lives on vertices. The triangles are equilateral with sides $k = 0,1,2$ as shown in Fig.~\ref{fig:triangles}. 
Each triangle $v$, at time $j$, hosts two $\mathbb{C}^3$ vectors, denoted by $(\psi_{j,k}^{0,v})_{k = 0,1,2}$ and $(\psi_{j,k}^{1,v})_{k=0,1,2}$. As each edge of the lattice is shared by two triangles, we label each triangle with $0$ or $1$ such that any two adjacent triangles have different labels, and we assign coin states $0$ to triangles labeled $0$ and coin states $1$ to triangles labeled $1$. The generic state of the walker at a given time $j'$ therefore reads:
\begin{equation}
    \Psi_{j'} = \sum_{v_0, k} \psi_{j',k}^{0,v_0} \ket{k_{v_0},0} + \sum_{v_1, k} \psi_{j',k}^{1,v_1} \ket{k_{v_1},1}
\end{equation}
where $v_0$ spans the $0$-labeled triangles, $v_1$ spans the $1$-labeled triangles and $k_v$ refers to the $k$-th edge of triangle $v$. The QW dynamics is recovered by the sequential application of two unitary local operators: the first operator $\hat C_k$, namely the quantum coin, is applied over each of the two-component wave-functions lying on edges labeled $k$, for $k = 0,1,2$. The coin depends on $k$ in general. The second operator, $R$, rotates every triangle anticlockwise, with component $k$ of each triangle hopping to component $(k+1) \bmod 3$, as shown in Fig.~\ref{fig:triangles}. Notice that it coincides with the simultaneous application of three shift operators $S_{u_k}$, along each unitary displacement vector $u_k,~k=0,1,2$, such that $R = \sum_{k=0}^2 S_{u_k}$, as shown in Fig.~\ref{fig:triangles}. Each operator $S_{u_k}$ acts on the edge $k$, internally to each triangle, moving the complex amplitude from the edge $k$ to the edge $k+1$. Let $\Psi^k = \sum_{v_0} \psi_k^{0,v_0} \ket{k_{v_0},0} + \sum_{v_1} \psi_k^{1,v_1} \ket{k_{v_1}, 1}$, we have, for each triangle $v_0$,
\begin{equation}
    \begin{pmatrix}
        (S_{u_k} \Psi^k)_{k+1 \bmod 3}^{0,v_0} \\
        (S_{u_k} \Psi^k)_{k+1 \bmod 3}^{1, e(k,v_0)}
    \end{pmatrix}
    =
    \begin{pmatrix}
        \psi_k^{0,v_0}\\
        \psi_k^{1,v_0}
    \end{pmatrix}
\end{equation}
where $e(k,v)$ denotes the neighbor of triangle $v$ along edge $k$. The overall shift operator can be written as follows~:
\begin{small}
\begin{eqnarray}
    S_{u_k} = \sum_{v_0} (\ket{(k+1 \bmod 3)_{v_0}, 0} \bra{k_{v_0}, 0} \nonumber \\ 
    + \ket{(k+1 \bmod 3)_{e(k, v_0)}, 1} \bra{k_{v_0}, 1}).
\end{eqnarray}
\end{small}

Let us introduce the position $r := (x,y)$ as the center of each edge $k_v$, for a given triangle $v$. The two-component wave function $\Psi_{j,k}^v$ reads~:
\begin{equation}
\Psi_{j,k}(r) = \begin{pmatrix} \psi_{j,k}^0(r) \\ \psi_{j,k}^1(r) \end{pmatrix}, 
\end{equation}
and for a fixed $k$, we therefore have

\begin{equation}
    \begin{pmatrix}
        (S_{u_k} \Psi_k)_{k+1 \bmod 3}^0(r) \\
        (S_{u_k} \Psi_k)_{k+1 \bmod 3}^1(r)
    \end{pmatrix}
    =
    \begin{pmatrix}
        \psi_k^0(r + \Delta u_k) \\
        \psi_k^1(r - \Delta u_k)
    \end{pmatrix}
    \label{eqn:shift_r}
\end{equation}

where $\Delta$ is the lattice discretization step of the triangulation, \emph{i.e.} half the length of an edge of the grid. It is also useful to relate the coordinate basis $\{u_k\},~k=0,1,2$ to the rectangular coordinate system $\{u_s\},~s=x,y$:
\begin{equation}
    u_k = \cos\left(\frac{2k\pi}3\right) u_x + \sin\left(\frac{2k\pi}3\right) u_y =: \mathcal{R}_k^s u_s
    \label{eqn:tri_basis}
\end{equation}
where $\mathcal{R}$ is the coordinates change matrix.

Altogether, the triangular QW is given by the following recursive relations~:
\begin{equation}
    \Psi_{j+1,k} = \left(\prod_{i=0,1,2} S_{u_{k+i \bmod 3}} \hat C_{k+i \bmod 3}\right) \Psi_{j,k}.
\end{equation}

\begin{figure}[hbt!]
\centering
\begin{tikzpicture}[scale=2]
\draw[black] (-1,0) -- (1,0) -- (0,{sqrt(3)}) -- cycle;
\filldraw[color=black,fill=black!30!white] (-.5,{sqrt(3)/2}) -- (0,0) -- (.5, {sqrt(3)/2}) -- cycle;
\node at (-.9,{sqrt(3)/4}) {1};
\node at (-.4,{sqrt(3)/2*.6}) {0};
\node at (-.5,-.15) {2};
\node at (.4,{sqrt(3)/2*.6}) {1};
\node at (.9,{sqrt(3)/4}) {0};
\node at (.5,-.15) {2};
\node at (-.4,{sqrt(3)*3/4}) {1};
\node at (.4,{sqrt(3)/2*1.5}) {0};
\node at (0,{sqrt(3)/2+.15}) {2};
\draw[gray,thin,->] (-.5,{sqrt(3)/2*.33}) +(-30:.2) arc(-30:270:.2);
\draw[blue,thick,->] (-.25,{sqrt(3)/4}) -- (.25, {sqrt(3)/4}) node[midway, below] {$u_0$};
\draw[red,thick,->] (.25,{sqrt(3)/4}) -- (0,{sqrt(3)/2}) node[midway, right] {$u_1$};
\draw[green!60!black,thick,->] (0,{sqrt(3)/2}) -- (-.25, {sqrt(3)/4}) node[midway,left] {$u_2$};
\end{tikzpicture}
\caption{The triangular QW: gray triangles are labeled 1, white ones are labeled 0}
\label{fig:triangles}
\end{figure}
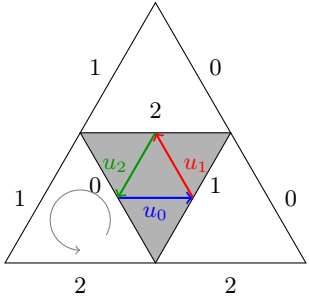

We will now prove that by changing the coin basis and choosing a specific coin, we recover in the continuous spacetime limit the massless Dirac equation in (2+1) dimensions.

\begin{theorem}[Dirac quantum walk over triangulation]
    Consider the previous quantum walk where:
    \begin{equation}
        \hat C_0 = \hat C_1 = \hat C_2 =: \hat C = \mathrm e^{i\pi/3} \mathcal{O}_z\left(-\frac{2\pi}3\right) = \mathrm e^{i\pi/3} \exp\left(i\frac{\pi}3 \sigma_z\right)
    \end{equation}
    If $\tilde \Psi_k(t, r) = U \Psi_k(t, r)$, $C' = U C U^\dagger$ and $U = \exp(-i\alpha \sigma_y/2) C^2$ with $\alpha = -\arccos(\sqrt5/3)$, then the quantum walk in the new basis reads as
    \begin{equation}
        \tilde \Psi_{j+1,k} = \left(\prod_{i=0,1,2} S_{u_{k+i \bmod 3}} UCU^\dagger\right) \tilde \Psi_{j,k}
        \label{eqn:qw_basischange}
    \end{equation}
    and, for $k = 0$, it admits as continuous limit in spacetime, with discretization parameters $\Delta = \varepsilon$ for the spatial dimension and $\Delta_t = \varepsilon$ for the time dimension, the following partial differential equation:
    \begin{equation}
        \partial_t \Psi(t,x,y) = (\sigma_x \partial_x + \sigma_y \partial_y) \Psi(t,x,y).
        \label{eqn:dirac_uniform}
    \end{equation}
\end{theorem}

\begin{proof}
By developing equation \ref{eqn:shift_r} around $\Delta$, we obtain:
\begin{equation}
    (S_{u_k} \Psi)(r) = \Psi(r) + \Delta \sigma_z \partial_{u_k} \Psi(r) + \mathrm o(\Delta)
\end{equation}
Therefore, by developing equation \ref{eqn:qw_basischange} around $\varepsilon$, we get up to the first order:
\begin{small}
\begin{multline}
    \tilde \Psi_0(t,r) + \varepsilon \partial_t \tilde\Psi_0(t,r) = C'^3 \tilde \Psi_0(t,r) + \varepsilon[C'^2 \sigma_z C' \partial_{u_0} \tilde \Psi_k(t,r) +\\
    +C' \sigma_z C'^2 \partial_{u_1} \tilde \Psi_k(t,r) + \sigma_z C'^3 \partial_{u_2} \tilde \Psi_0(t,r)] + \mathrm o(\varepsilon)
\end{multline}
\end{small}
Then by reversing the basis change:
\begin{multline}
    \Psi_0(t,r) + \varepsilon \partial_t \Psi_0(t,r) = U^\dagger C'^3 U \Psi_0(t,r)\\
    + \varepsilon(\tau_2 \partial_{u_2} + \tau_1 \partial_{u_1} + \tau_0 \partial_{u_0}) \Psi_0(t,r) + \mathrm o(\varepsilon)
    \label{eqn:dl_uniform}
\end{multline}
where:
\begin{equation}
    \begin{split}
        \tau_0 &= C^2 U^\dagger \sigma_z U C \\
        \tau_1 &= C U^\dagger \sigma_z U C^2 \\
        \tau_2 &= U^\dagger \sigma_z U
    \end{split}
    \label{eqn:tau_exprs}
\end{equation}

We have $U^\dagger C'^3 U = U^\dagger (U C U^\dagger)^3 U = C^3 = I_2$ with our choice of $C$. To prove the theorem, we therefore need:

\begin{equation}
    \sum_{i=0}^2 \tau_i \partial_{u_i} = \sigma_x \partial_x + \sigma_y \partial_y.
    \label{eqn:partialder_condition}
\end{equation}

We first translate the $(\partial_{u_i})_{i=0,1,2}$ in terms of the coordinates $(x,y)$, using equation \ref{eqn:tri_basis}:
\begin{equation}
    \partial_{u_k} = \cos\left(\frac{2i\pi}3\right) \partial_x + \sin\left(\frac{2i\pi}3\right) \partial_y = \mathcal{R}_k^s \partial_s
\end{equation}

We can then derive a relation between the $\tau_i$ and the $\sigma^s$:
\begin{equation}
    \mathcal{R}_i^s \tau^i = \sigma^s,\quad s=x,y 
    \label{eqn:tausigma_expr}
\end{equation}

This equation, along with the unitarity and traceless condition on the $\tau_i$, leads to the following unique $\tau_i$ matrices, up to a sign:
\begin{equation}
    \begin{split}
        \tau_0 &= \frac23 \sigma_x + \kappa \sigma_z \\
        \tau_1 &= -\frac13 \sigma_x + \frac{\sqrt3}3 \sigma_y + \kappa \sigma_z \\
        \tau_2 &= -\frac13 \sigma_x - \frac{\sqrt3}3 \sigma_y + \kappa \sigma_z
    \end{split}
    \label{eqn:tau_sol}
\end{equation}
with $\kappa = \pm \frac{\sqrt5}3$. By choosing $\kappa = \frac{\sqrt5}3$, and by using equations \ref{eqn:tau_exprs} and \ref{eqn:tau_sol}: 
\begin{equation}
    \tau_0 = C^2 U^\dagger \sigma_z U C = \frac23 \sigma_x + \kappa \sigma_z,
    \label{eqn:tau_combined}
\end{equation} we can derive $UC = \exp(-i\alpha \sigma_y/2)$ with $\alpha = -\arccos(\sqrt5/3)$, hence $$U = \exp(-i\alpha \sigma_y/2) C^\dagger = \exp(-i\alpha \sigma_y/2) C^2$$ 
where in the last sentence we used that $C$ is unitary and $C^3 = I_2$ from the zero$^{th}$-order condition. Finally notice that with the aforementioned choice of $U$ and $C$, the equations \ref{eqn:tau_exprs}, \ref{eqn:tau_sol} and \ref{eqn:tau_combined} are satisfied. Therefore, by taking the formal limit $\varepsilon \to 0$ in equation \ref{eqn:dl_uniform} and using equation \ref{eqn:partialder_condition}, we recover the following Dirac equation in (2+1) dimensions:
\begin{equation}
    \partial_t \Psi(t,x,y) = (\sigma_x \partial_x + \sigma_y \partial_y) \Psi(t,x,y).
\end{equation}
\end{proof}
We simulated the Dirac QW with different initial states. In each density plot, the color at a given point corresponds to the sum of the probability densities at each edge of the triangle at this point. All simulations were performed on grids of $802 \times 400$ triangles, with periodic conditions. The spatial coordinates shown on the axes are multiples of $\varepsilon$. The source for C++ code for all the simulations included in this manuscript is freely accessible at \url{https://github.com/vdng9338/qw_simul}.
In Fig.~\ref{fig:unif_square0}.a, whose initial state is concentrated in a single point, the walker propagates isotropically at unit speed. In Fig.~\ref{fig:unif_square0}.b, where the initial state is more spread, we can see a left-moving behavior along with vertical spreading (still at unit speed). 

\begin{figure}[hbt!]
    \centering
    \begin{minipage}{.23\textwidth}
        \centering
        \includegraphics[width=\textwidth,clip,trim={50px 50px 50px 50px}]{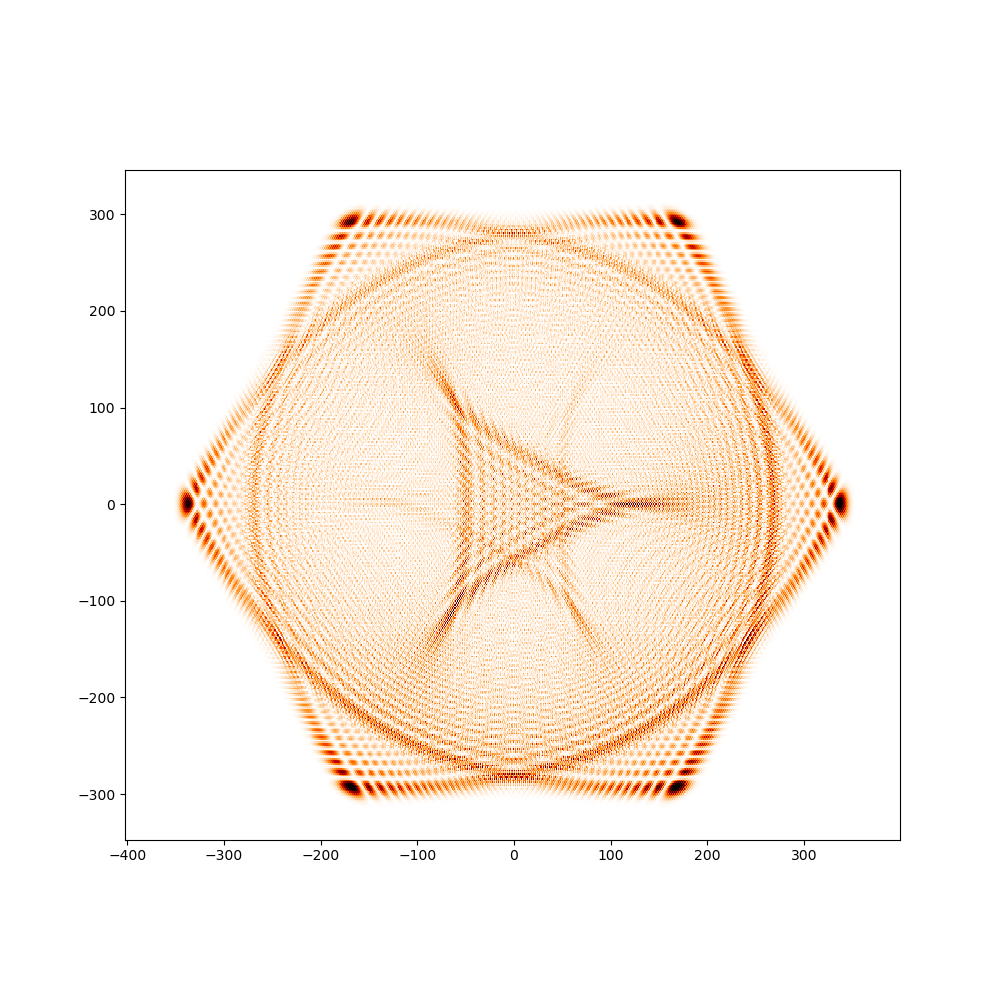}
        \subcaption{After 300 steps}
    \end{minipage}
    \begin{minipage}{.23\textwidth}
        \centering
        \includegraphics[width=\textwidth,clip,trim={50px 50px 50px 50px}]{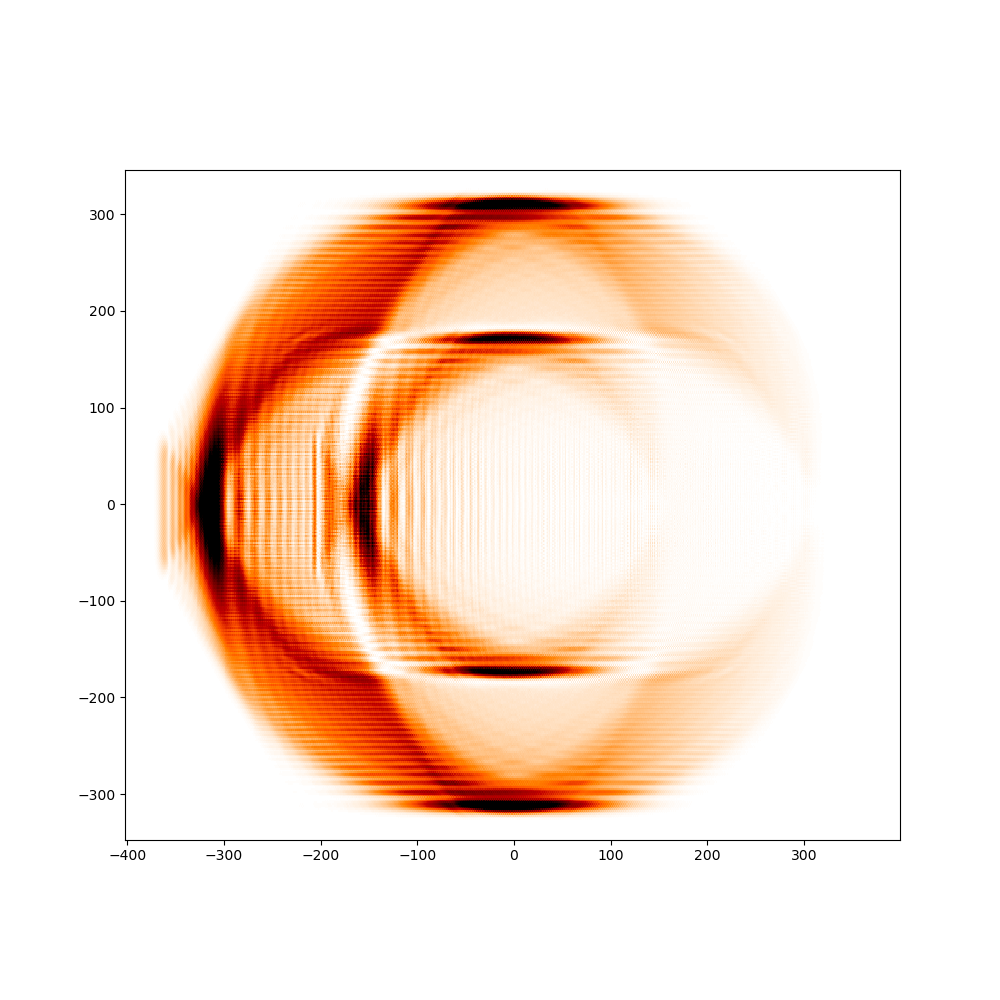}
        \subcaption{After 250 steps}
    \end{minipage}
    \caption{(Color online) (a) Simulation with initial amplitudes $1$ on edge $0$ of triangle $(0,0)$; (b) Simulation with an initially uniformly spread amplitude over a centered rectangle (size of the rectangle: a fifth of each dimension), sides $0$ only}
    \label{fig:unif_square0}
\end{figure}

\section{The Geodesic Quantum Walk}
\label{sec:geoQW}

Simulating/modelling transport on physical molecules such as fullerenes requires taking into account the non-flat geometry of such structures. Let us consider, as an example, the triangulation of the sphere. A Dirac QW can be implemented on it, as shown in Fig. \ref{fig:sphere}. 
Indeed, to recover the metrics of a unit sphere, we (conformally) project it onto a plane at $z~=~-1$, and we determine local orthonormal bases of the sphere that project onto orthogonal bases of the plane aligned with the axes.
The inverse of the stereographical projection reads:
\begin{small}
\begin{equation}
    \varphi(x,y) = \left(\frac{2x}{\frac12(x^2+y^2)+2}, \frac{2y}{\frac12(x^2+y^2)+2}, 1-\frac{4}{\frac12(x^2+y^2)+2}\right),
\end{equation}
\end{small}
and its partial derivatives read:
\begin{small}
\begin{equation}
    \begin{split}
        \frac{\partial \varphi}{\partial x} &= \left( \frac{4(-x^2+y^2+4)}{(x^2+y^2+4)^2}, -\frac{8xy}{(x^2+y^2+4)^2}, \frac{16x}{(x^2+y^2+4)^2}\right) \\
        \frac{\partial \varphi}{\partial y} &= \left( -\frac{8xy}{(x^2+y^2+4)^2}, \frac{4(x^2-y^2+4)}{(x^2+y^2+4)^2}, \frac{16y}{(x^2+y^2+4)^2}\right).
    \end{split}
\end{equation}
\end{small}
Flattening a part of the sphere on a plane corresponds to deforming the local basis of the triangulation by~: 
\begin{small}
\begin{equation}
    \Lambda(x,y) = \begin{pmatrix}
        1/\lVert \frac{\partial \varphi}{\partial x} \rVert_2 & 0 \\
        0 & 1/\lVert \frac{\partial \varphi}{\partial y} \rVert_2
    \end{pmatrix}.
\end{equation}
\end{small}
Once the metrics is flattened, but deformed, we can implement the very same QW, introduced previously. We also know that, if the metrics is diagonal, as in the aforementioned example, and we are sufficiently far from the singularities~\footnote{Indeed the full way to model the discrete curvature, including the singularities, is by considering the presence of dislocations, such as pentagonal/conical defects (e.g. 12 exactly for fullerenes)~\cite{topofull}. A systematic study of dislocations in the context of QWs has been extensively considered by one of the authors in an independent work.}, the curved transport of the walker can be simulated over a regular (aka non deformed) triangulation, at the price of considering non uniform evolution operators. In the following we will prove that there exists a more general and new family of QW, that we call Geodesic QW (GQW), capable to simulate \textit{any} arbitrary deformed triangulation. Moreover we introduce anisotropic spacetime discretization in order to recover the most general curved Dirac equation in the continuous limit.


\begin{figure}[t]
\includegraphics[width=10cm]{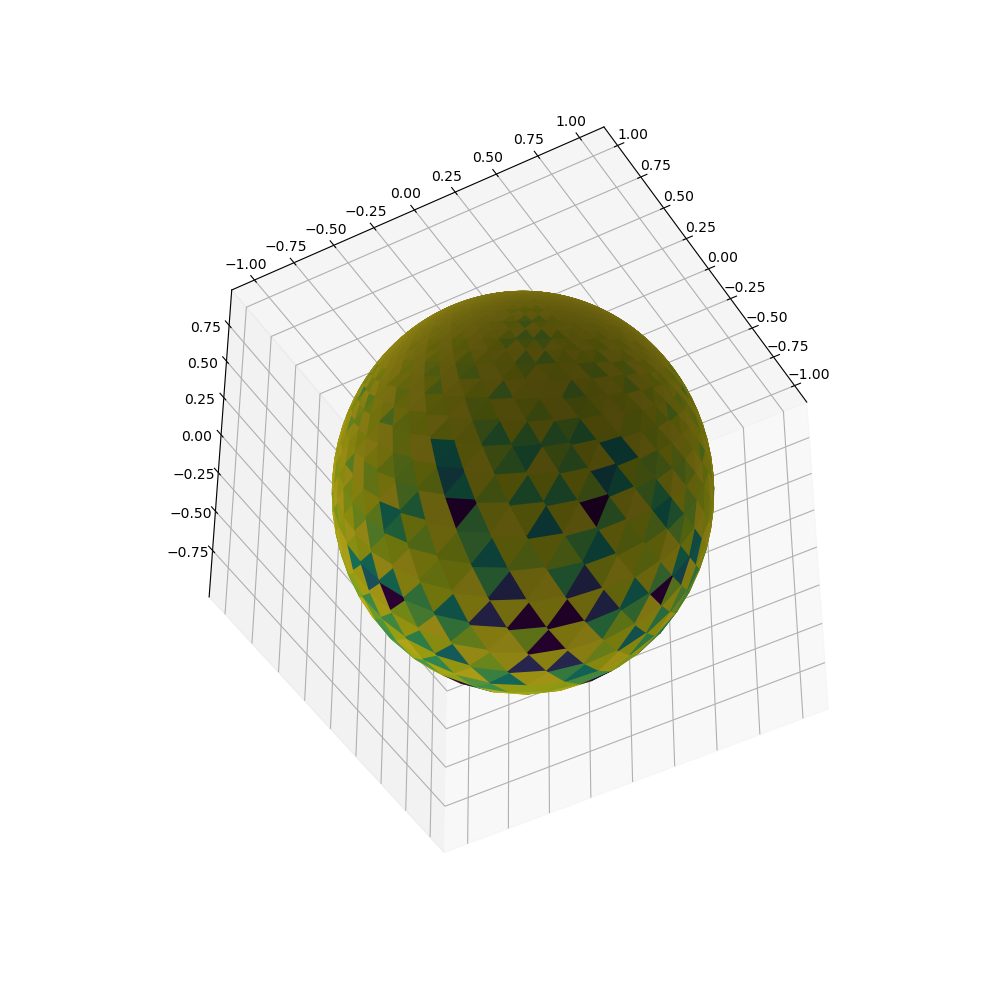}
\centering
\caption{(Color online) Dirac Quantum Walker over the stereographically projected sphere after 600 steps.} 
 \label{fig:sphere}
\end{figure}

To begin with, let us explore the case of a global homogeneous deformation of the lattice. Such a transformation consists in a uniform change of basis of the following form:
\begin{equation}
    u'_i = \begin{pmatrix} \lambda_{00} & \lambda_{01} \\ \lambda_{10} & \lambda_{11} \end{pmatrix} u_i =: \Lambda u_i
\end{equation}
meaning that $u'_x = \lambda_{00} u_x + \lambda_{10} u_y$ and $u'_y = \lambda_{01} u_x + \lambda_{11} u_y$.

The derivatives $\partial_{u_i}$ also transform as:
\begin{equation}
    \partial_{u'_i} = \Lambda \partial_{u_i}.
    \label{eqn:uniform_der}
\end{equation}

By actually deforming the triangulation and using the Dirac QW defined in the previous section, that is by solving together \ref{eqn:uniform_der} and \ref{eqn:dirac_uniform}, we get:
\begin{equation}
  \partial_t \psi = \{B^s,\partial_s\} \psi
    \label{eqn:unif_deform_pde}
\end{equation}
with
\begin{equation}
    \begin{split}
    B^x &= \lambda_{00} \sigma^x + \lambda_{01} \sigma^y \\
    B^y &= \lambda_{10} \sigma^x + \lambda_{11} \sigma^y
    \end{split}
\end{equation}
or, with equation \ref{eqn:tri_basis},
\begin{equation}
    B^s = \Lambda_i^s \sigma^i = \Lambda_i^s R_i^j \tau_j.
    \label{eqn:b_expr}
\end{equation}
where $\Lambda_i^s = \lambda_{si}$, $0$ corresponds to $x$ and $1$ corresponds to $y$. Then from equation \ref{eqn:dl_uniform}, we get:
\begin{equation}
    \partial_t \Psi_k(t,r) = (\tau'_2 \partial_{u_2} + \tau'_1 \partial_{u_1} + \tau'_0 \partial_{u_0})
\end{equation}
with:
\begin{equation}
    \begin{split}
        \tau'_0 &= \frac23 (\lambda_{00} \sigma_x + \lambda_{01} \sigma_y) + \kappa \sigma_z \\
        \tau'_1 &= -\frac13 (\lambda_{00} \sigma_x + \lambda_{01} \sigma_y) + \frac{\sqrt3}3 (\lambda_{10} \sigma_x + \lambda_{11} \sigma_y) + \kappa \sigma_z \\
        \tau'_2 &= -\frac13 (\lambda_{00} \sigma_x + \lambda_{01} \sigma_y) - \frac{\sqrt3}3 (\lambda_{10} \sigma_x + \lambda_{11} \sigma_y) + \kappa \sigma_z,
    \end{split}
\end{equation}
that is, we obtain the matrices $\tau'_i$ by applying the deformation $\Lambda$ to the $\sigma_x$ and $\sigma_y$ coordinates of the matrices $\tau_i$.

The velocity field remains uniform here. Now, in order to simulate an inhomogeneous velocity field, we need to choose a spacetime-dependent $\Lambda(t,x,y)$ transformation. Instead of introducing this distortion on the lattice via the modification of the $u_i$ vectors, the unitary matrices $\tau_i$ can be transformed into matrices $\beta_i(t,x,y)$ to produce the same effect. We indeed seek for a set of matrices $\beta^i(t,x,y)$ that fullfil the following conditions:
\begin{itemize}
    \item (C1) We impose that \begin{equation}
        \Lambda_k^j(t,x,y) R_i^k \tau^i = R_i^j \beta^i(t,x,y), \quad j=x,y
    \end{equation}
    We obtain this condition by equating expressions \ref{eqn:b_expr} and \ref{eqn:tausigma_expr} (with $\tau^i$ replaced by $\beta^i$ in the latter).
    \item (C2) Each of the $\beta^i$ has $\{-1,1\}$ as eigenvalues, i.e. at any time step and any point $(x,y)$ of the lattice, there exist three unitaries $U_i(t,x,y)$ such that
    \begin{equation}
        \beta^i(t,x,y) = U_i^\dagger(t,x,y) \sigma_z U_i(t,x,y).
    \end{equation}
\end{itemize}

Note that condition (C1), that we call the duality condition, implies that the coordinate transformation dictated by $\Lambda_k^j(t,x,y)$ is transferred to the unitary operations $\beta^i(t,x,y)$, instead of the original $\tau^i$. Additionally, condition (C2) will allow us to rewrite the QW evolution in terms of the usual state-dependent translation operators.

To lighten the notations, we will omit the spacetime dependence both in these matrices and in the $U_i(t,x,y)$. The above conditions allow to calculate the $\beta^i$ matrices, which can be written as a combination of Pauli matrices, \emph{i.e.} $\beta^i = \vec n^i \cdot \vec \sigma$, where each $\vec n^i$ must be a real, unit vector $\vec n^i = (\sin \theta_i, 0, \cos \theta_i)$ for some spacetime-dependent angles $\theta_i$. Moreover, we will see that the $\sigma_y$ component is not needed to achieve our purpose.

In this way,
\begin{equation}
    \beta_i = U_i^\dagger \sigma_z U_i =
    \begin{pmatrix}
        \cos \theta_i & \sin \theta_i \\
        \sin \theta_i & -\cos \theta_i
    \end{pmatrix}
\end{equation}
and by diagonalizing each $\beta_i$, we obtain
\begin{equation}
    U_i(\theta_i) = \begin{pmatrix}
        \cos \frac{\theta_i}2 & \sin \frac{\theta_i}2 \\
        -\sin \frac{\theta_i}2 & \cos \frac{\theta_i}2
    \end{pmatrix}.
\end{equation}

The most naive way to implement such a walk is by alternating the unitaries $U_i$ and the shift operators as seen in the first section, where the unitaries were homogeneous due to an appropriate choice of the coin state basis and parameters. However, the unitaries will be three different ones, depending on three independent real parameters $\theta_i$, which is not sufficient to simulate all possible deformations. In fact, each spacetime dependent $\Lambda$ depends on four real free parameters. Thus, in order to recover this local deformation, we need at least one more internal parameter in the QW evolution. Moreover, terms like $\partial_j \lambda_{ij}$ are necessary to conserve the probability distribution of the walker. This can be achieved by iterating twice each unitary operator. All these considerations lead us to define the Geodesic Quantum Walk operator as follows: 
\begin{equation}
    \Psi_{j+2,k} = Z_2 Z_1 \Psi_{j,k}
    \label{eqn:curved_qw}
\end{equation}
where
\begin{equation}
    \begin{split}
        Z_1 &= H \left(\prod_{i=0,1,2} \bar V_{k+i \bmod 3}\right) \left(\prod_{i=0,1,2} V_{k+i \bmod 3}\right) H \\
        Z_2 &= Q^\dagger \left(\prod_{i=0,1,2} \bar K_{k+i \bmod 3}\right) \left(\prod_{i=0,1,2} K_{k+i \bmod 3}\right) Q
    \end{split}
\end{equation}
and
\begin{equation}
    \begin{split}
        V_i &= U_i(\theta_i) S_{u_i} U_i(\theta_i)^\dagger \\
        \bar V_i &= U_i(\theta_i)^\dagger S_{u_i} U_i(\theta_i) \\
        K_i &= U_{i+3}(\theta_{i+3}) S_{u_i} U_{i+3}(\theta_{i+3})^\dagger \\
        \bar K_i &= U_{i+3}(\theta_{i+3})^\dagger S_{u_i} U_{i+3}(\theta_{i+3})
    \end{split}
\end{equation}
where
\begin{equation}
    H = \frac1{\sqrt2} \begin{pmatrix} 1 & 1 \\ 1 & -1 \end{pmatrix} \quad
    Q = \frac1{\sqrt2} \begin{pmatrix} 1 & -i \\ 1 & i \end{pmatrix}
\end{equation}
In practice, for edges labeled $k$, one step of the walk consists in first applying $H$ on each edge, then applying $U_k(\theta_k(t,x,y))^\dagger$ on each edge, then shifting along $u_k$, then applying $U_k(\theta_k(t,x,y))$ on each \emph{$(k+1)$-labeled} edge, then applying $U_{k+1}(\theta_{k+1}(t,x,y))^\dagger$ on each $(k+1)$-labeled-edge, and so on.

Let us discuss this choice. Each $Z_i$ iterates two modified versions of the quantum walk seen in Sec.~\ref{subsect:triangular_qw}, choosing different unitaries for each edge $k$ of the triangle in order to get in the continuous limit spatial derivatives of the unitaries $U_i$. Notice that for the second iteration, we have chosen the transpose conjugate of the unitaries $U_i$. This is justified by the fact that, in the end, we wish to recover spatial derivatives of the form $\partial_j \beta_i = (\partial_j U_i^\dagger) \sigma_z U_i + U_i^\dagger \sigma_z \partial_j U_i$. Iterating twice the same operator $V_i$ or $K_i$ would not be sufficient to recover the total derivative, and we would have twice $U_i^\dagger \sigma_z \partial_j U_i$. Finally, $H$ and $Q$ provide the change of basis in the coin state basis to recover in the continuous limit a true 2D propagation, similarly to the change of basis $U$ in Sec.~\ref{subsect:triangular_qw}. In conclusion, two of the $Z_i$ are necessary to have enough free parameters for the deformation $\Lambda$.

Let us now prove that the GQW correctly converges to the Dirac equation in a $(2+1)$-curved spacetime. To this scope we introduce the following anisotropic scaling following~\cite{manighalam2021continuous}:
\begin{equation}
    \theta_i(t,x,y) = \frac\pi2 + \varepsilon^{1/2} l_i(t,x,y) \quad \Delta_x = \varepsilon^{1/2} \quad \Delta_t = \varepsilon
\end{equation}
for some spacetime-dependent functions $l_i$.
By expanding equation \ref{eqn:curved_qw} up to first order in $\varepsilon$, after a tedious but straightforward computation, which we spare the reader from detailing here, one arrives to the following equation in the continuum limit $\varepsilon \to 0$:
\begin{small}
\begin{multline}
    \partial_t \Psi = [\partial_{u_2} (l_2 \sigma_x + l_5 \sigma_y) + \partial_{u_1} (l_1 \sigma_x + l_4 \sigma_y) + \partial_{u_0} (l_0 \sigma_x + l_3 \sigma_y)] \Psi\\
    + [(l_2 \sigma_x + l_5 \sigma_y) \partial_{u_2} + (l_1 \sigma_x + l_4 \sigma_y) \partial_{u_1} + (l_0 \sigma_x + l_3 \sigma_y) \partial_{u_0}] \Psi
\end{multline}
\end{small}

Now, using equation \ref{eqn:uniform_der}, we can reformulate the above equation in terms of $\partial_x$ and $\partial_y$ and link the coefficients $l_i(r)$ and $\lambda_{ij}(r)$:
\begin{multline}
    \partial_t \Psi = \partial_x(\lambda_{00} \sigma_x + \lambda_{01} \sigma_y) \Psi + (\lambda_{00} \sigma_x + \lambda_{01} \sigma_y) \partial_x \Psi \\
    + \partial_y(\lambda_{10} \sigma_x + \lambda_{11} \sigma_y) \Psi + (\lambda_{10} \sigma_x + \lambda_{11} \sigma_y) \partial_y \Psi
\end{multline}
where
\begin{equation}
    \begin{split}
        \lambda_{00} &= -\frac12(l_2+l_1)+l_0 \\
        \lambda_{01} &= -\frac12(l_5+l_4)+l_3 \\
        \lambda_{10} &= -\frac{\sqrt3}2(l_2-l_1) \\
        \lambda_{11} &= -\frac{\sqrt3}2(l_5-l_4)
    \end{split}
\end{equation}
which coincides with the curved Dirac equation, in its canonical form~\cite{Sinha1994}. Notice that we simply recover Eq.~\ref{eqn:unif_deform_pde} for homogeneous $\Lambda$. Moreover, observe that we have a system of linear equations which is overdetermined, which leaves us enough freedom to recover the deformation matrix we wish. For instance, a good choice to gauge away this ambiguity is $l_5 = -l_4$ and $l_2 = -l_1$, which leads to the unique choice:
\begin{equation}
    l_0 = \lambda_{00} \quad l_1 = \frac{\lambda_{10}}{\sqrt3} \quad l_3 = \lambda_{01} \quad l_4 = \frac{\lambda_{11}}{\sqrt3}
\end{equation}
We have thus proven that a non-homogeneous arbitrary deformation of the triangulation can be simulated by local non-homogeneous unitaries applied on the edges of a homogeneous regular triangulation.\\

\section{Conclusion}\label{sec:conc}
A new and more general family of QWs on triangulations has been introduced. We have shown that it is possible to relate quantum transport on curved surfaces to a quantum simulator, on regular triangular grids, based on a non-uniform spatial distribution of local unitaries. GQWs are a non-trivial generalization of the well-known geodesic random walkers, where the metric is embedded in the displacement operators. We believe that this result has relevance for both modeling and simulating quantum transport on discrete structures, such as fullerene molecules. Moreover we aim, in the next future, to address fast and efficient optimization problems by using quantum curved space optimization methods, inspired by the aforementioned results. 

\bibliographystyle{ieeetr} 
\bibliography{bibliography}

\begin{thebibliography}{10}

\bibitem{arrighi2019curved}
P.~Arrighi, G.~Di~Molfetta, I.~Marquez-Martin, and A.~Perez, ``From curved
  spacetime to spacetime-dependent local unitaries over the honeycomb and
  triangular quantum walks,'' {\em Scientific reports}, vol.~9, no.~1,
  pp.~1--10, 2019.

\bibitem{knight1962random}
F.~B. Knight, ``On the random walk and brownian motion,'' {\em Transactions of
  the American Mathematical Society}, vol.~103, no.~2, pp.~218--228, 1962.

\bibitem{roberts1960random}
P.~H. Roberts and H.~D. Ursell, ``Random walk on a sphere and on a riemannian
  manifold,'' {\em Philosophical Transactions of the Royal Society of London.
  Series A, Mathematical and Physical Sciences}, vol.~252, no.~1012,
  pp.~317--356, 1960.

\bibitem{varopoulos1984brownian}
N.~T. Varopoulos, ``Brownian motion and random walks on manifolds,'' in {\em
  Annales de l'institut Fourier}, vol.~34, pp.~243--269, 1984.

\bibitem{arrighi2018dirac}
P.~Arrighi, G.~Di~Molfetta, I.~M{\'a}rquez-Mart{\'\i}n, and A.~P{\'e}rez,
  ``Dirac equation as a quantum walk over the honeycomb and triangular
  lattices,'' {\em Physical Review A}, vol.~97, no.~6, p.~062111, 2018.

\bibitem{marevs2020quantum}
J.~Mare{\v{s}}, J.~Novotn{\`y}, and I.~Jex, ``Quantum walk transport on carbon
  nanotube structures,'' {\em Physics Letters A}, vol.~384, no.~15, p.~126302,
  2020.

\bibitem{verga2019interacting}
A.~D. Verga, ``Interacting quantum walk on a graph,'' {\em Physical Review E},
  vol.~99, no.~1, p.~012127, 2019.

\bibitem{matsue2016quantum}
K.~Matsue, O.~Ogurisu, and E.~Segawa, ``Quantum walks on simplicial
  complexes,'' {\em Quantum Information Processing}, vol.~15, no.~5,
  pp.~1865--1896, 2016.

\bibitem{aristote2020dynamical}
Q.~Aristote, N.~Eon, and G.~Di~Molfetta, ``Dynamical triangulation induced by
  quantum walk,'' {\em Symmetry}, vol.~12, no.~1, p.~128, 2020.

\bibitem{arrighi2014dirac}
P.~Arrighi, V.~Nesme, and M.~Forets, ``The dirac equation as a quantum walk:
  higher dimensions, observational convergence,'' {\em Journal of Physics A:
  Mathematical and Theoretical}, vol.~47, no.~46, p.~465302, 2014.

\bibitem{arrighi2018quantum}
P.~Arrighi, G.~Di~Molfetta, and S.~Facchini, ``Quantum walking in curved
  spacetime: discrete metric,'' {\em Quantum}, vol.~2, p.~84, 2018.

\bibitem{di2013quantum}
G.~Di~Molfetta, M.~Brachet, and F.~Debbasch, ``Quantum walks as massless dirac
  fermions in curved space-time,'' {\em Physical Review A}, vol.~88, no.~4,
  p.~042301, 2013.

\bibitem{arnault2017quantum}
P.~Arnault and F.~Debbasch, ``Quantum walks and gravitational waves,'' {\em
  Annals of Physics}, vol.~383, pp.~645--661, 2017.

\bibitem{arrighi2017quantum}
P.~Arrighi and F.~Facchini, ``Quantum walking in curved spacetime:(3+ 1)
  dimensions, and beyond,'' {\em Quantum Information \& Computation}, vol.~17,
  no.~9-10, pp.~0810--0824, 2017.

\bibitem{jorgensen1975central}
E.~J{\o}rgensen, ``The central limit problem for geodesic random walks,'' {\em
  Zeitschrift f{\"u}r Wahrscheinlichkeitstheorie und verwandte Gebiete},
  vol.~32, no.~1, pp.~1--64, 1975.

\bibitem{ma2021geodesic}
T.~Ma, V.~S. Matveev, and I.~Pavlyukevich, ``Geodesic random walks, diffusion
  processes and brownian motion on finsler manifolds,'' {\em The Journal of
  Geometric Analysis}, vol.~31, no.~12, pp.~12446--12484, 2021.

\bibitem{novoselov2005two}
K.~S. Novoselov, A.~K. Geim, S.~V. Morozov, D.~Jiang, M.~I. Katsnelson,
  I.~Grigorieva, S.~Dubonos, and a.~Firsov, ``Two-dimensional gas of massless
  dirac fermions in graphene,'' {\em nature}, vol.~438, no.~7065, pp.~197--200,
  2005.

\bibitem{kitagawa2010exploring}
T.~Kitagawa, M.~S. Rudner, E.~Berg, and E.~Demler, ``Exploring topological
  phases with quantum walks,'' {\em Physical Review A}, vol.~82, no.~3,
  p.~033429, 2010.

\bibitem{ambjorn2014quantum}
J.~Ambj{\o}rn, A.~G{\"o}rlich, J.~Jurkiewicz, and R.~Loll, ``Quantum gravity
  via causal dynamical triangulations,'' in {\em Springer handbook of
  spacetime}, pp.~723--741, Springer, 2014.

\bibitem{rovelli2015covariant}
C.~Rovelli and F.~Vidotto, {\em Covariant loop quantum gravity: an elementary
  introduction to quantum gravity and spinfoam theory}.
\newblock Cambridge University Press, 2015.

\bibitem{duranthon2021coarse}
O.~Duranthon and G.~Di~Molfetta, ``Coarse-grained quantum cellular automata,''
  {\em Physical Review A}, vol.~103, no.~3, p.~032224, 2021.

\bibitem{Note1}
Indeed the full way to model the discrete curvature, including the
  singularities, is by considering the presence of dislocations, such as
  pentagonal/conical defects (e.g. 12 exactly for fullerenes)~\cite {topofull}.
  A systematic study of dislocations in the context of QWs has been extensively
  considered by one of the authors in an independent work.

\bibitem{manighalam2021continuous}
M.~Manighalam and G.~Di~Molfetta, ``Continuous time limit of the dtqw in 2d+ 1
  and plasticity,'' {\em Quantum Information Processing}, vol.~20, no.~2,
  pp.~1--24, 2021.

\bibitem{Sinha1994}
A.~Sinha and R.~Roychoudhury, ``Dirac equation in (1+1)-dimensional curved
  space-time,'' {\em International Journal of Theoretical Physics}, vol.~33,
  pp.~1511--1522, 07 1994.

\end{thebibliography}
\end{document}